\documentclass[11pt]{article}%
\usepackage{amsfonts}
\usepackage{amsmath}%
\usepackage{amssymb}%
\usepackage{graphicx}
\usepackage{mathtools}

\usepackage[colorlinks=true,linkcolor=blue,citecolor=red,plainpages=false,pdfpagelabels]%
{hyperref}%

\providecommand{\U}[1]{\protect\rule{.1in}{.1in}}
\newtheorem{theorem}{Theorem}

\newtheorem{lemma}[theorem]{Lemma}

\newtheorem{remark}[theorem]{Remark}

\newenvironment{proof}[1][Proof]{\noindent\textbf{#1.} }{\ \rule{0.5em}{0.5em}}

\numberwithin{equation}{section}

\allowdisplaybreaks

\begin{document}

\title{Wave Matrix Lindbladization I: Quantum Programs for Simulating Markovian Dynamics}
\author{Dhrumil Patel\thanks{Department of Computer Science, Cornell University, Ithaca, New York 14850, USA, Email: djp265@cornell.edu} \and Mark M. Wilde\thanks{School of Electrical and Computer Engineering, Cornell University, Ithaca, New York 14850, USA, Email: wilde@cornell.edu}}
\maketitle

\begin{abstract}Density Matrix Exponentiation is a technique for simulating Hamiltonian dynamics when the  Hamiltonian to be simulated is available as a quantum state. In this paper, we present a natural analogue to this technique, for simulating Markovian dynamics governed by the well known Lindblad master equation. For this purpose, we first propose an input model in which a Lindblad operator $L$ is encoded into a quantum state $\psi$. Then, given access to $n$ copies of the state $\psi$, the task is to simulate the corresponding Markovian dynamics for time $t$. We propose a quantum algorithm for this task, called Wave Matrix Lindbladization, and we also investigate its sample complexity. We show that our algorithm uses $n = O(t^2/\varepsilon)$ samples of $\psi$ to achieve the target dynamics, with an approximation error of~$O(\varepsilon)$.
\end{abstract}

\bigskip
\begin{quote}
\textit{We dedicate our paper to the memory of G\"oran Lindblad (July~9, 1940--November~30, 2022), whose profound contributions to quantum information science, in the form of the Lindblad master equation~\cite{Lindblad1976OnSemigroups} and the data-processing inequality for quantum relative entropy \cite{Lin75}, will never be forgotten.}
\end{quote}

\newpage
\tableofcontents

\section{Introduction}

\subsection{Background}

By the early 1980s, it was evident that simulating the behavior of complex quantum systems using a classical computer is computationally expensive, since it involves keeping track of an exponentially large number of quantum state amplitudes. 
In order to overcome this difficulty, Feynman proposed the idea of simulating quantum systems using a computational device that is itself quantum mechanical \cite{Feynman1982SimulatingComputers}. 
Originally, it was merely a conjecture; however, it gave birth to the field of quantum simulation \cite{Lloyd1996UniversalSimulators,Georgescu2013QuantumSimulation}, which is currently one of the most anticipated applications of a quantum computer. 
The essential idea behind quantum simulation involves using a quantum computer to perform the simulation, which can then allow for a detailed investigation of the quantum system being simulated.

Hamiltonian simulation is a particular kind of quantum simulation that involves simulating the behavior of a closed quantum system.
This is an extensively investigated area, and many quantum algorithms have been developed to date to solve this specific problem \cite{Lloyd1996UniversalSimulators,Berry2013ExponentialHamiltonians,Berry2012Gate-efficientAlgorithms,Lloyd2014QuantumAnalysis,Berry2014SimulatingSeries,Berry2015HamiltonianParameters,Low2016OptimalProcessing,Kimmel2017HamiltonianComplexity}. 

While Hamiltonian simulation is a well studied problem, it is limited to only closed quantum systems. 
In many real-world scenarios, quantum systems are subject to the influence of an environment, leading to more complex dynamics that are better described by open system models. 
Moreover, if the dynamics of an open system are Markovian in nature (i.e., its quantum state at time $t+ \Delta$ only depends on the quantum state at time $t$ and is independent of states before time $t$), then such dynamics are well captured by the Lindblad master equation.
The general form of this equation was delineated independently by Göran Lindblad \cite{Lindblad1976OnSemigroups} and by Gorini, Kossakowski, and Sudarshan \cite{Gorini2008CompletelySystems}, and so, this equation is also known as the Gorini--Kossakowski--Sudarshan--Lindblad equation.
The significance of this master equation cannot be overstated.
It is crucial in understanding the behavior of a wide range of quantum systems and scenarios \cite{breuer2002theory,Weiss2021}, including condensed matter \cite{Prosen2011OpenTransport,Manzano2011QuantumLaw,Olmos2012FacilitatedGlasses}, quantum chemistry \cite{nitzan2006chemical,may2008charge},  quantum optics \cite{Plenio1998,Gardiner2004QuantumBooks}, entanglement preparation \cite{Kraus2008,Kastoryano2011,Reiter2016}, 
 thermal state preparation \cite{Kastoryano2014QuantumCase}, quantum state engineering \cite{verstraete2009quantum},
and the effects of noise on quantum computers~\cite{Magesan2012ModelingCircuits}. 

In this paper, we consider the problem of simulating the Lindbladian evolution of a finite-dimensional quantum system in an initial state $\rho$ for time~$t$. This evolution is governed by the following Lindblad master equation:
\begin{equation}
\label{eq:lindbladmaster}
    \frac{\partial \rho}{\partial t} =  \mathcal{L} (\rho) \coloneqq  -i [H, \rho] + \sum_{k=1}^{K} L_{k}\rho L_{k}^{\dagger} - \frac{1}{2} \left \{L_{k}^{\dagger}L_{k}, \rho \right\},
\end{equation}
where $H$ is a Hermitian operator representing the system's Hamiltonian, and the operators $\{L_{k}\}_{k=1}^{K}$ are called Lindblad operators, which are not necessarily Hermitian and in fact have no constraints on them. In addition, the superoperator $\mathcal{L}$ is known as a Lindbladian. The notation $\{A, B\}$ above refers to the anti-commutator of operators $A$ and $B$, i.e., $\{A, B\} = AB + BA$. By simulating the aforementioned evolution for time $t$, we mean implementing its corresponding quantum channel~$e^{\mathcal{L}t}$, which is the solution of \eqref{eq:lindbladmaster}, where
\begin{equation}
    e^{\mathcal{L}t}(\rho) = \sum_{k=0}^\infty \frac{\mathcal{L}^k(\rho)t^k}{k!} ,
    \label{eq:Lind-expand}
\end{equation}
and $\mathcal{L}^k$ denotes $k$ sequential applications of the Lindbladian $\mathcal{L}$.
For small $t$, note that $e^{\mathcal{L}t}(\rho) = \rho +  \mathcal{L}(\rho) t + O(t^2)$, and we make use of this expansion in what follows.

Throughout our paper, we focus on a simple case in which the Lindblad master equation consists of only a single Lindblad operator $L$. For clarity, we rewrite the Lindbladian corresponding to this simple case:
\begin{equation}\label{eq:lindblad-master-single}
    \mathcal{L} (\rho) = L\rho L^{\dagger} - \frac{1}{2} \left \{L^{\dagger}L, \rho \right\}.
\end{equation}
We use this basic scenario as a starting point, as it is easier to grasp the intuition behind the techniques we introduce here. 
Furthermore, one can easily extend this case to simulate more complex Lindbladian evolutions with multiple Lindblad operators, by using Proposition~2 of \cite{Childs2016EfficientDynamics}. Specifically, this proposition states that, given efficient implementations of polynomially many Lindbladians $\mathcal{L}_{1}, \mathcal{L}_{2}, \ldots, \mathcal{L}_{m}$, one can efficiently implement their linear combination $\sum_{i=1}^{m} \mathcal{L}_{i}$.

More recently, there has been growing interest in developing efficient quantum algorithms for simulating the dynamics of open quantum systems, as given by \eqref{eq:lindbladmaster} \cite{Childs2016EfficientDynamics,Cleve2016EfficientEvolution,KSMM22,Schlimgen2022QuantumOperators,Suri2022Two-UnitarySimulation} (see \cite{Miessen2022QuantumDynamics} for a review). These works are primarily based on the assumption that some succinct representation of the Lindblad operators or black-box access to them are provided beforehand. 
For example, a list of non-zero coefficients when writing these operators as a linear combination of Paulis is one such succinct representation~\cite{Cleve2016EfficientEvolution}.

In our paper, we approach the above problem from a different angle. We assume that the Lindblad operator $L$ is encoded in a pure quantum state~$|\psi\rangle$, and we have access to multiple copies of this state. That is, we suppose that~$L$ is encoded  in $|\psi\rangle$ in the following manner:
\begin{equation}
    |\psi\rangle \coloneqq (L\otimes I) |\Gamma\rangle,
    \label{eq:program-state-Lindblad-L}
\end{equation}
where $|\Gamma\rangle \coloneqq \sum_{j} |j\rangle |j\rangle $ is a maximally entangled vector. This way of encoding lies at the heart of our quantum algorithm, and as far as we are aware, it is the first time that such an encoding scheme has been proposed. We refer to such a state as a program state, as it can be programmed to encode any square linear operator according to the problem at hand. The only constraint on the operator $L$, encoded as above, is that $\left\Vert L\right\Vert_{2} = 1$, where $\left\Vert A \right\Vert_{2} \coloneqq \sqrt{\operatorname{Tr}[A^{\dagger}A]}$ is the Schatten-2 norm of a matrix $A$ (also known as the Hilbert--Schmidt norm). This constraint on $L$ arises from the fact that $|\psi\rangle$ is a quantum state. That being said, for encoding a Lindblad operator $L'$ with an arbitrary norm and corresponding Lindbladian $\mathcal{L}'$, we can suppose that its normalized version, i.e., $L'/\left\Vert L' \right\Vert_{2}$, is encoded in a quantum state. Then, for simulating its corresponding quantum channel $e^{\mathcal{L}'t}$, we simulate the channel $e^{\mathcal{L}' t'/\left\Vert L' \right\Vert_{2}^2 }$ for time $t' = \left\Vert L' \right\Vert_{2}^2 t$, so that $e^{\mathcal{L}'t}=e^{\mathcal{L}' t'/\left\Vert L' \right\Vert_{2}^2 }$. This is evident from \eqref{eq:lindblad-master-single}. Thus, without loss of generality,  we can assume that the Lindblad operator $L$ is normalized, i.e., $\left\Vert L\right\Vert_{2} = 1$, and we do so throughout our paper.

We refer to this newly introduced method of Lindbladian simulation as \textit{Wave Matrix Lindbladization}. The reasoning behind this terminology is that~$L$ is known as a wave matrix \cite{McCaul2023TheDynamics}, and we are ``lindbladizing" it, i.e., transforming it from a wave matrix into a Lindblad operator. Essentially, we ask: given one copy of an unknown quantum state $\rho$ and $n$ copies of the program state $\psi \coloneqq |\psi\rangle\! \langle\psi|$, can we approximately implement the quantum channel $e^{\mathcal{L}t}$ up to an approximation error $\varepsilon$? That is, can we realize the following transformation?
\begin{equation}\label{eq:task-map}
    \rho \otimes \underbrace{\psi \otimes \cdots \otimes \psi}_{\text{$n$ times}} \overset{\overset{\varepsilon}{\approx}}{\longrightarrow} e^{\mathcal{L}t}(\rho).
\end{equation}

It is worth noting that Wave Matrix Lindbladization can be seen as a natural analogue to \textit{Density Matrix Exponentiation} \cite{Lloyd2014QuantumAnalysis}. Density Matrix Exponentiation is a well known protocol for Hamiltonian simulation, and it is also used in the context of quantum machine learning \cite{Biamonte2016QuantumLearning}. The task here is to implement a unitary $e^{-i\rho t}$ given multiple copies of an unknown quantum state $\rho$, so that the quantum state $\rho$ serves as a Hamiltonian in this case.
 Similarly, in Wave Matrix Lindbladization, given multiple copies of an unknown quantum state $\psi$ encoding an operator $L$, the task is to ``lindbladize" this operator, i.e., implement the transformation given by~\eqref{eq:task-map}.

\subsection{Summary of Main Results}

In this paper, we propose a quantum algorithm that implements the quantum channel $e^{\mathcal{L}t}$ with some desired accuracy $\varepsilon$, where $0 < \varepsilon < 1$.
We then investigate the sample complexity used by our algorithm (see Section~\ref{sec:single-op-with-no-H}). 
By sample complexity, we mean the number of copies of the program state~$\psi$ used to achieve the above task. 
Furthermore, we extend the single Lindblad operator case to the case in which the Lindblad master equation also consists of a Hamiltonian term (see Section~\ref{sec:single-op-with-H}), and 
we propose a quantum algorithm for this case as well.

\textbf{Key Idea --- }Our quantum algorithm primarily involves two steps, which we repeat for each copy of the program state $\psi$. Suppose that $\rho$ is in register~1 and the program state $\psi$ is in registers 2 and 3. In the first step, we evolve~$\rho$ and $\psi$ according to the following Lindbladian $\mathcal{M}$ for a short duration of time~$\Delta \coloneqq t/n$:
\begin{equation}\label{eq:intro-M}
        \mathcal{M}\left (\rho \otimes  \psi\right) \coloneqq M(\rho \otimes  \psi) M^{\dagger} - \frac{1}{2} \left \{M^{\dagger}M, \rho \otimes  \psi \right\},
\end{equation}
where $M$ is a Lindblad operator defined as
\begin{equation}
        M \coloneqq \frac{1}{\sqrt{d}}\left(I_1\otimes |\Gamma\rangle\! \langle \Gamma |_{23}\right) \left( \mathsf{SWAP}_{12} \otimes I_3\right).
        \label{eq:M-op-initial}
\end{equation}
Here, $\mathsf{SWAP}_{12}$ is the swap operation that swaps the states in registers 1 and~2, and $|\Gamma\rangle$ is the maximally entangled vector. We explicitly define them both later in \eqref{eq:max-ent-vec} and \eqref{eq:swap-def}, respectively. 
Note that by evolving according to \eqref{eq:intro-M} for time~$\Delta$, this leads to the application of the quantum channel~$e^{\mathcal{M}\Delta}$. 
Furthermore, the second step of our algorithm involves tracing out the program state $\psi$. To see this algorithm pictorially, please refer to Figure~\ref{fig:n-times}.

\begin{figure}[t]
    \centering
    \includegraphics[scale=0.19]{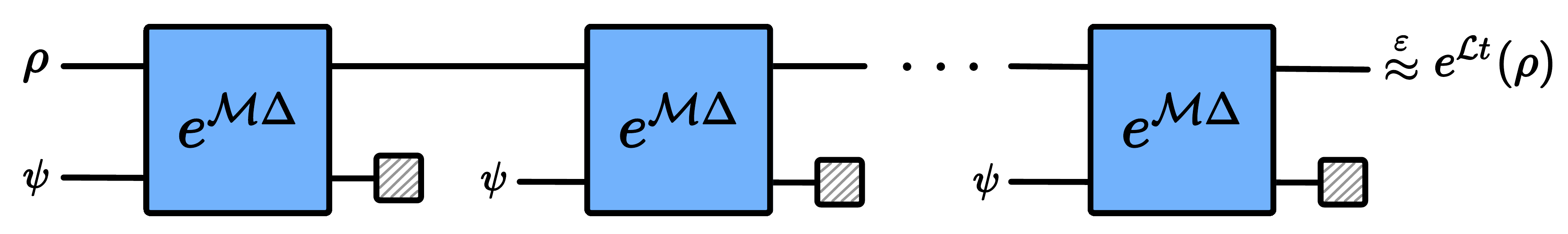}
    \caption{Our quantum algorithm, repeated $n = O(t^2/\varepsilon)$ times, for approximating the target quantum channel, i.e., $e^{\mathcal{L}t}$, with an approximation error of~$\varepsilon$. Each small, hatched square represents the trace-out operation. At the very left, the state $\rho$ is in register~1, and the program state $\psi$ is in registers~2 and~3, represented as a single line at this instance and thereafter for simplicity.}
    \label{fig:n-times}
\end{figure}

Formally, we can write the above two steps as
\begin{equation}
\operatorname{Tr}_{23}\!\left[e^{\mathcal{M}\Delta}(\rho \otimes \psi)\right] = \rho + \operatorname{Tr}_{23}\!\left[\mathcal{M}(\rho \otimes \psi) \right] \Delta + O(\Delta^2),
\end{equation}
where the notation $\operatorname{Tr}_{23}[\cdot]$ is used to denote the action of tracing out registers 2 and 3. The above equality follows from a Taylor series expansion, as discussed just after \eqref{eq:Lind-expand}. Now, the critical step here is to prove that
\begin{equation}
\label{eq:intro-cs}
    \operatorname{Tr}_{23}\!\left[\mathcal{M}(\rho \otimes \psi) \right] = \mathcal{L}(\rho).
\end{equation}
We prove this equality in detail in Section~\ref{sec:single-op-with-no-H} and Appendix~\ref{app:key-lemma}. The equality in \eqref{eq:intro-cs} is important because it implies the following:
\begin{align}
\operatorname{Tr}_{23}\!\left[e^{\mathcal{M}\Delta}(\rho \otimes \psi)\right] & = \rho + \operatorname{Tr}_{23}\!\left[\mathcal{M}(\rho \otimes \psi) \right] \Delta + O(\Delta^2)\\
& = \rho + \mathcal{L}(\rho) \Delta + O(\Delta^2) \\
& = e^{\mathcal{L}\Delta}(\rho) + O(\Delta^2).
\end{align}

If we repeat these two steps $n = O(t^2/\varepsilon)$ times, then we can approximate the target channel, i.e., $e^{\mathcal{L}t}$, with an approximation error of $O(\varepsilon)$. We prove this statement in Theorem~\ref{thm:single-op}. This theorem in turn  is a consequence of the key lemma of this paper (Lemma~\ref{lemma:key-lemma} in Appendix~\ref{app:key-lemma}), which establishes the equality in \eqref{eq:intro-cs}, as well as the error analysis in Appendix~\ref{app:error-analysis}.

\subsection{Notation}

We use the notation $\mathcal{H}_{S}$ to denote a $d$-dimensional Hilbert space associated with a quantum system $S$. 
We denote the set of quantum states acting on $\mathcal{H}_{S}$ by~$\mathcal{D}(\mathcal{H}_{S})$. 
Let $\operatorname{Tr}[X]$ denote the trace of a matrix $X$, i.e., the sum of its diagonal elements. 
Also, let $X^{\dagger}$ denote the Hermitian conjugate (or adjoint) of the matrix $X$.
The Schatten $p$-norm of a matrix $X$ is defined for $p \in [1, \infty)$ as follows:
\begin{equation}
    \left\Vert X\right\Vert_{p} \coloneqq \left(\operatorname{Tr}\!\left[\left(X^{\dagger}X\right)^{\frac{p}{2}}\right]\right)^{\frac{1}{p}}.
\end{equation}
For the purpose of this paper, we use Schatten norms with $p =1$ (also called trace norm) and $p=2$ (Hilbert--Schmidt norm). Furthermore, let  $[X, Y] \coloneqq XY - YX$ and $\{X, Y\} \coloneqq XY + YX$ denote the commutator and anti-commutator of the operators $X$ and~$Y$, respectively. 

The diamond distance between two quantum channels $\mathcal{N}$ and $\mathcal{M}$ is defined as follows \cite{Kitaev1997QuantumCorrection}:
\begin{equation}
     \left\Vert \mathcal{N} - \mathcal{M} \right\Vert_{\diamond} \coloneqq \sup_{\rho \in \mathcal{D}(\mathcal{H}_{R} \otimes \mathcal{H}_{S})}  \left \Vert (\mathcal{I}_{R} \otimes \mathcal{N}  )(\rho) - (\mathcal{I}_{R} \otimes \mathcal{M}  )(\rho)  \right \Vert_{1},
\end{equation}
where $R$ is a reference system and $\mathcal{I}_{R}$ is the identity channel acting on the system $R$. An important point to note here is that, in the above definition, the dimension of $R$ is arbitrarily large. However, it is known that it suffices to perform the optimization over pure bipartite states with the dimension of~$R$ equal to the dimension of $S$. Furthermore, the quantity in the objective function of the above optimization is the trace distance, defined as $\left \Vert \rho -\sigma \right\Vert_{1}$ for two quantum states $\rho, \sigma \in \mathcal{D}(\mathcal{H_{S}})$. In what follows, we employ the normalized diamond distance $\frac{1}{2} \left\Vert \mathcal{N} - \mathcal{M} \right\Vert_{\diamond}$ to measure approximation error---the normalization factor of $\frac{1}{2}$ guarantees that $\frac{1}{2} \left\Vert \mathcal{N} - \mathcal{M} \right\Vert_{\diamond} \in [0,1]$ for quantum channels $\mathcal{N}$ and $\mathcal{M}$.

\begin{figure}[t]
    \centering
    \includegraphics[scale=0.31]{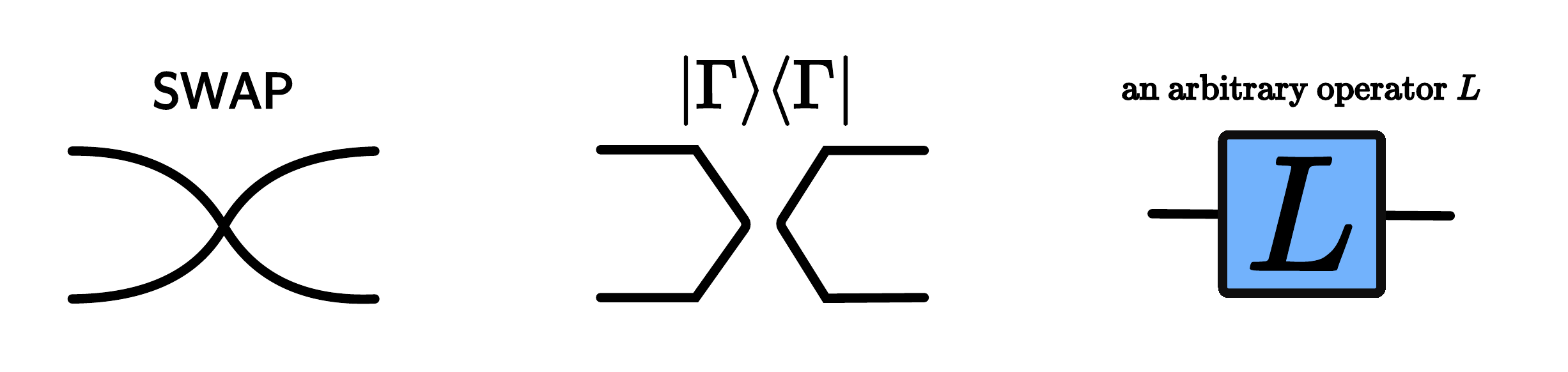}
    \caption{Tensor-network diagrams of operators $\mathsf{SWAP}$, $|\Gamma\rangle\!\langle\Gamma|$, and $L$.}
    \label{fig:lindblad-symb}
\end{figure}

For $\rho$ a quantum state in $\mathcal{D}(\mathcal{H}_{R} \otimes \mathcal{H}_{S})$, we denote the partial trace over the Hilbert space $\mathcal{H}_{R}$ by $\operatorname{Tr}_{R}[\rho]$. We also sometimes use a different notation for partial trace; i.e., given a multi-partite state $\rho$, we use the notation $\operatorname{Tr}_{k}[\rho]$ to denote the action of tracing out the $k^{\text{th}}$ party. Furthermore, we define the maximally entangled vector in $\mathcal{H}_{R} \otimes \mathcal{H}_{S}$ as
\begin{equation}
|\Gamma\rangle_{RS} \coloneqq \sum_{i} |i\rangle_{R} |i\rangle_{S}.    
\label{eq:max-ent-vec}
\end{equation}
 We also define the unitary swap operation in the following way:
\begin{equation}
    \mathsf{SWAP} \coloneqq \sum_{i, j} |i\rangle\! \langle j| \otimes |j\rangle\! \langle i|.
    \label{eq:swap-def}
\end{equation}

In our paper, we make extensive use of tensor-network diagrams. Figure~\ref{fig:lindblad-symb} depicts tensor-network diagrams for some basic operators defined above, such as $\mathsf{SWAP}$ and $|\Gamma\rangle\!\langle\Gamma|$. For more background on tensor-network diagrams, please refer to \cite{Biamonte2017TensorNutshell}. Throughout this paper, we sometimes suppress system labels for ease of notation; however, they will be clear from the context.

\section{Quantum Algorithms for Simulating Markovian Dynamics}

\subsection{Single-Operator Case With No Hamiltonian Term}\label{sec:single-op-with-no-H}

In this section, we provide a detailed analysis of our quantum algorithm for simulating the quantum channel $e^{\mathcal{L}t}$, in the case that the Lindbladian has only one Lindblad operator $L$, as in~\eqref{eq:lindblad-master-single}. The algorithm simulates the channel~$e^{\mathcal{L}t}$ up to error $\varepsilon$ in normalized diamond distance, using $n$ copies of the program state $\psi$ that encodes $L$ (recall \eqref{eq:program-state-Lindblad-L} here).
Since we are interested in implementing the aforementioned channel with some desired accuracy in diamond distance, we assume that the channel input state is a joint quantum state of two systems rather than just one.
Therefore, let $\rho \in \mathcal{D}(\mathcal{H}_{R} \otimes \mathcal{H}_{S})$ be an unknown quantum state given as input over the joint system $RS$, where the system $R$ is a reference system. Furthermore, let the $k^{\text{th}}$ copy of the program state $\psi$ be a quantum state of the joint system $P_{k}Q_{k}$.

\textbf{Algorithm~1 ---} Set $n\in \mathbb{N}$, with the particular choice specified later. Set $k=1$. Given the $k^{\text{th}}$ copy of $\psi$, i.e., $\psi_{P_{k}Q_{k}}$, perform the following two steps:\label{algo:single-op}
\begin{enumerate}
    \item Evolve the joint quantum state $\rho_{RS}\otimes \psi_{P_{k}Q_{k}}$ according to the dynamics realized by the following Lindbladian $\mathcal{M}$, for some small duration of time~$\Delta = t/n $:
    \begin{equation}\label{eq:single-fixed-L-master}
        \mathcal{M}\left (\rho_{RS} \otimes  \psi_{P_{k}Q_{k}}\right) \coloneqq M(\rho_{RS} \otimes  \psi_{P_{k}Q_{k}}) M^{\dagger} - \frac{1}{2} \left \{M^{\dagger}M, \rho_{RS} \otimes  \psi_{P_{k}Q_{k}} \right\}.
    \end{equation}
    In the above, the Lindblad operator $M$ acts on the joint system $RSP_{k}Q_{k}$, and we define it as
    \begin{equation}
        M \coloneqq \frac{1}{\sqrt{d}}\left(I_{RS}\otimes |\Gamma\rangle\! \langle \Gamma |_{P_{k}Q_{k}}\right) \left( I_{R}\otimes\mathsf{SWAP}_{SP_{k}} \otimes I_{Q_{k}}\right).
        \label{eq:orig-M-def}
    \end{equation}
    (Note that we have redefined $M$ as compared to \eqref{eq:M-op-initial}, in order to include the trivial action on the reference system $R$.)
    \item Trace out the systems $P_{k}Q_{k}$.
\end{enumerate}
We repeat the above procedure using each copy of $\psi$, i.e., for all $k$ ranging from $1$ to $n$.

The following theorem states that the above algorithm uses $n = O(t^2/\varepsilon)$ copies of $\psi$ to simulate the Lindbladian evolution of $\rho_{RS}$, given by~\eqref{eq:Lind-expand}--\eqref{eq:lindblad-master-single}, for time $t$, such that the final state is $\varepsilon$-close in normalized trace distance to the ideal target state $\left(\mathcal{I}_{R} \otimes e^{\mathcal{L}t} \right)\left(\rho_{RS}\right)$, for an arbitrary input state $\rho_{RS}$.

\begin{theorem}\label{thm:single-op}
Given access to $n$ copies of the program state $\psi \in  \mathcal{D}(\mathcal{H}_{P} \otimes \mathcal{H}_{Q})$, which encodes the Lindblad operator $L$ as in \eqref{eq:program-state-Lindblad-L}, there exists a quantum algorithm~$\mathcal{A}$ such that the following error bound holds:
    \begin{equation}
        \frac{1}{2}\left \Vert e^{\mathcal{L}t} - \mathcal{A} \right \Vert_{\diamond} \leq \varepsilon,
    \end{equation}
    with only $n = O(t^2/\varepsilon)$ copies of $\psi$. In other words, $\mathcal{A}$ uses only $n = O(t^2/\varepsilon)$ copies of $\psi$ to approximate the channel $ e^{\mathcal{L}t}$ up to $\varepsilon$ error in normalized diamond distance.
\end{theorem}
 \begin{proof}
In what follows, we provide a brief sketch of the proof. For a more detailed version of the proof, please refer to Appendix~\ref{app:error-analysis}. For ease of notation and simplicity, we refrain from writing the system labels, and we also assume that the input state $\rho$ does not have the reference system $R$ for the purpose of this proof sketch.

Let us begin by expanding the target state $e^{\mathcal{L}t}(\rho)$ using the following Taylor series expansion, as in \eqref{eq:Lind-expand}, at the initial time of $t=0$:
\begin{equation}\label{eq:single-desired-state}
    e^{\mathcal{L}t}(\rho) = \rho + \mathcal{L}(\rho) t +  \frac{1}{2} (\mathcal{L}\circ  \mathcal{L})(\rho) t^2 + \ldots .
\end{equation}
In the first step of Algorithm~1, we simulate the Lindbladian evolution of $\rho \otimes \psi$, using the Lindbladian $\mathcal{M}$ in \eqref{eq:single-fixed-L-master} for some small duration of time $\Delta$, and then trace out $\psi$. The output state obtained after this step is
\begin{equation}\label{eq:single-op-proof-sketch-evo}
    \operatorname{Tr}_{23}\!\left[e^{\mathcal{M}\Delta}(\rho \otimes \psi)\right] = \rho + \operatorname{Tr}_{23}\!\left[\mathcal{M}(\rho \otimes \psi) \right] \Delta + O(\Delta^2),
\end{equation}
where $\rho$ is in register~1 and $\psi$ in registers 2 and 3, and we have again used the expansion in \eqref{eq:Lind-expand}.
Then writing out the second term on the right-hand side of the above equation and using the definition in~\eqref{eq:single-fixed-L-master}, we find that
\begin{multline}
    \operatorname{Tr}_{23}\!\left[\mathcal{M}(\rho \otimes \psi) \right] = \operatorname{Tr}_{23}\!\left [M(\rho \otimes \psi) M^{\dagger} \right] \\
    - \frac{1}{2}\operatorname{Tr}_{23}\!
    \left [M^{\dagger}M \left(\rho \otimes  \psi\right) \right] -  \frac{1}{2}\operatorname{Tr}_{23}\!
    \left [\left(\rho \otimes  \psi \right)M^{\dagger}M  \right].
\end{multline}
We then invoke Lemma~\ref{lemma:key-lemma} in Appendix~\ref{app:key-lemma} to simplify each term on the right-hand side of the above equation. As a result of this, we obtain the following equalities:
\begin{align}
    \operatorname{Tr}_{23}\!\left [M(\rho \otimes \psi) M^{\dagger} \right] & = L\rho L^{\dagger},\\
    \operatorname{Tr}_{23}\!
    \left [M^{\dagger}M \left(\rho \otimes  \psi \right) \right] & =  L^{\dagger}L \rho , \\
    \operatorname{Tr}_{23}\!
    \left [\left(\rho \otimes  \psi\right)M^{\dagger}M  \right] & =   \rho L^{\dagger}L.
\end{align}
For a graphical representation of the above simplifications, please refer to the tensor-network diagrams provided in Figures~\ref{fig:lindblad-first-term-no-H}, \ref{fig:lindblad-second-term-no-H}, and \ref{fig:lindblad-third-term-no-H}.

Using the above equations along with \eqref{eq:lindblad-master-single}, we rewrite \eqref{eq:single-op-proof-sketch-evo} as
\begin{align}
    \operatorname{Tr}_{23}\!\left[e^{\mathcal{M}\Delta}(\rho \otimes \psi)\right] & = \rho + \mathcal{L}(\rho)\Delta + O(\Delta^2) \\
    & = e^{\mathcal{L}\Delta}(\rho) + O(\Delta^2).
\end{align}

Substituting $\Delta = t/n$ and repeating Algorithm~1 for $n = O(t^2 / \varepsilon)$ times produces a quantum state that is $O(\varepsilon)$-close to the ideal target state $e^{\mathcal{L}t}(\rho)$ in normalized trace distance. For a detailed error analysis of this claim, in terms of the diamond distance, please refer to Appendix~\ref{app:error-analysis}.
\end{proof}


\begin{remark}
\label{rem:M-choice}
We note here that the simulation is unchanged if we employ the following definition of $M$, instead of that given in \eqref{eq:orig-M-def}:
\begin{equation}
        M \coloneqq \left(I_{RS}\otimes |\varphi\rangle\! \langle \Gamma |_{P_{k}Q_{k}}\right) \left( I_{R}\otimes\mathsf{SWAP}_{SP_{k}} \otimes I_{Q_{k}}\right).
\end{equation}
where $|\varphi\rangle$ is an arbitrary bipartite state vector. As such, the choice of~$|\varphi\rangle$ in \eqref{eq:orig-M-def} amounts to the maximally entangled state $\frac{1}{\sqrt{d}}|\Gamma\rangle$. The claim here can be checked by examining Figures~\ref{fig:lindblad-first-term-no-H}, \ref{fig:lindblad-subroutine-no-H}, \ref{fig:lindblad-second-term-no-H}, and \ref{fig:lindblad-third-term-no-H}, as well as the proof in Appendix~\ref{app:key-lemma}.
\end{remark}

\begin{figure}
    \centering
    \includegraphics[scale=0.14]{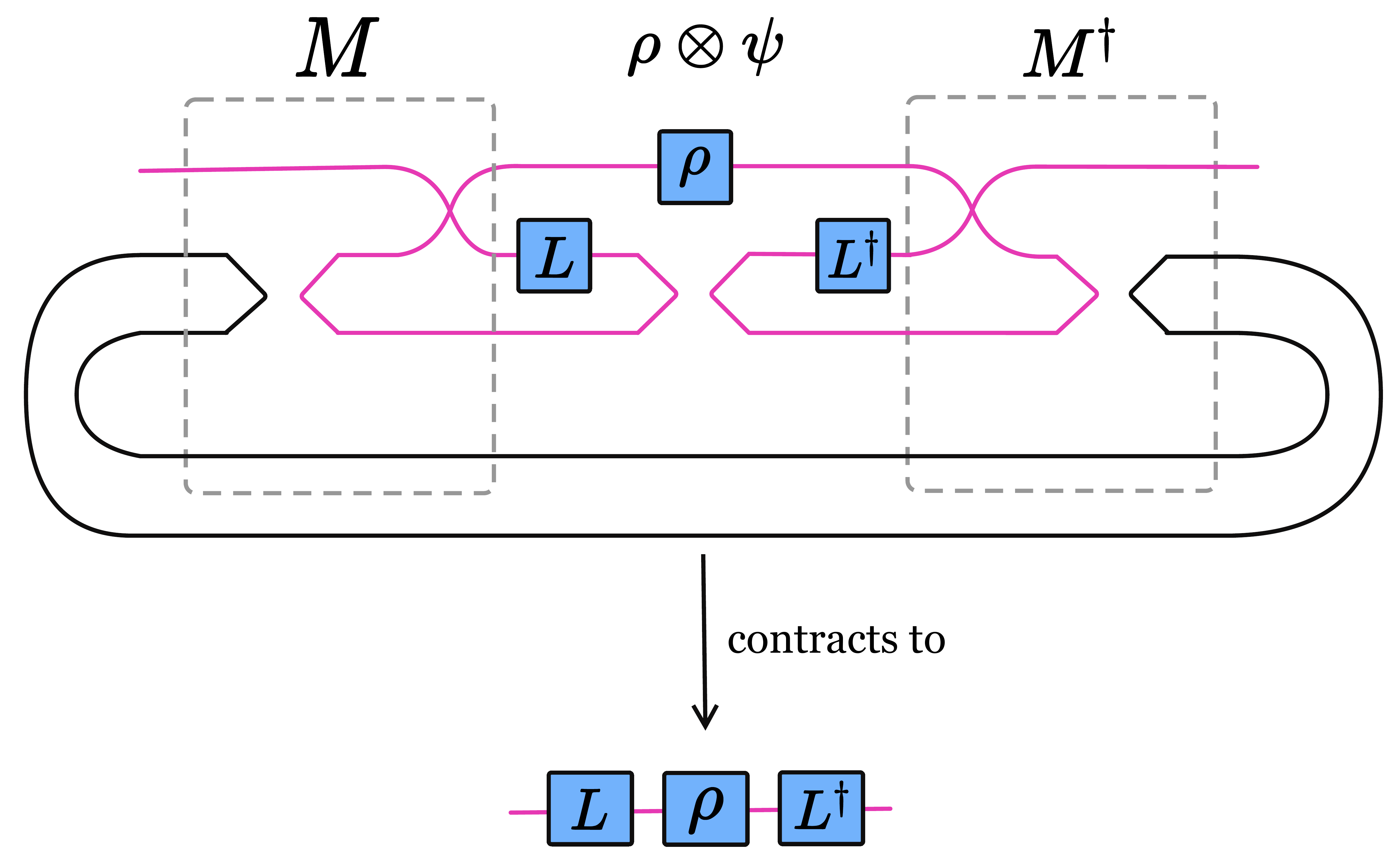}
    \caption{Tensor-network diagram of $\operatorname{Tr}_{23}\!\left[M(\rho \otimes \psi) M^{\dagger} \right]$. The whole network on the top contracts to the network on the bottom due to the partial trace operation over the $2^{\text{nd}}$ and $3^{\text{rd}}$ systems. Simply follow the pink line from the left of the first system to its right to observe this.}
    \label{fig:lindblad-first-term-no-H}
\end{figure}
\begin{figure}
    \centering
    \includegraphics[scale=0.16]{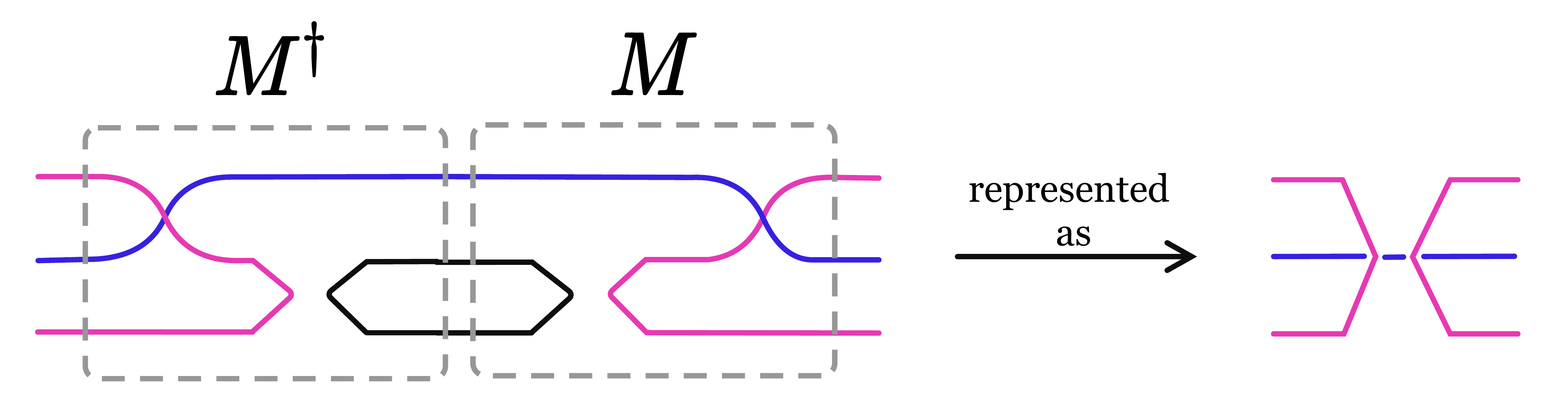}
    \caption{Tensor-network diagram of the operator $M^{\dagger}M$. This network will be used as a subroutine in the following figures. As a result, for the sake of brevity, we use the figure on the right to represent the figure on the left.}
    \label{fig:lindblad-subroutine-no-H}
\end{figure}
\begin{figure}
    \centering
    \includegraphics[scale=0.13]{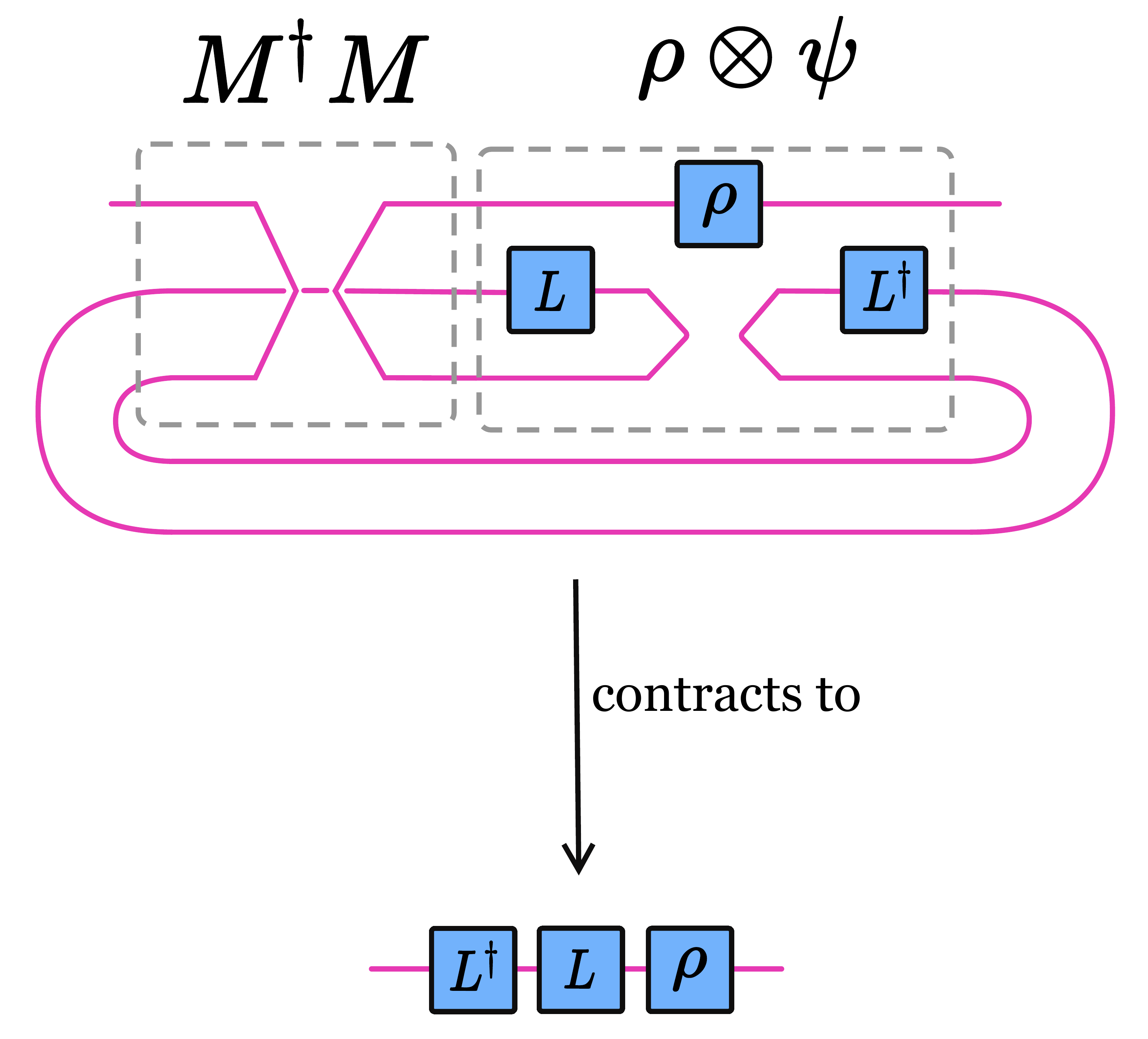}
    \caption{Tensor-network diagram of $\operatorname{Tr}_{23}\!
    \left [M^{\dagger}M \left(\rho \otimes  \psi \right) \right]$. Please refer to Figure~\ref{fig:lindblad-subroutine-no-H} to understand the tensor-network diagram of $M^{\dagger}M$. The whole network on the top contracts to the network on the bottom due to the partial trace operation over the $2^{\text{nd}}$ and $3^{\text{rd}}$ systems. Simply follow the pink line from left of the first system to its right to observe this.}
    \label{fig:lindblad-second-term-no-H}
\end{figure}
\begin{figure}
    \centering
    \includegraphics[scale=0.13]{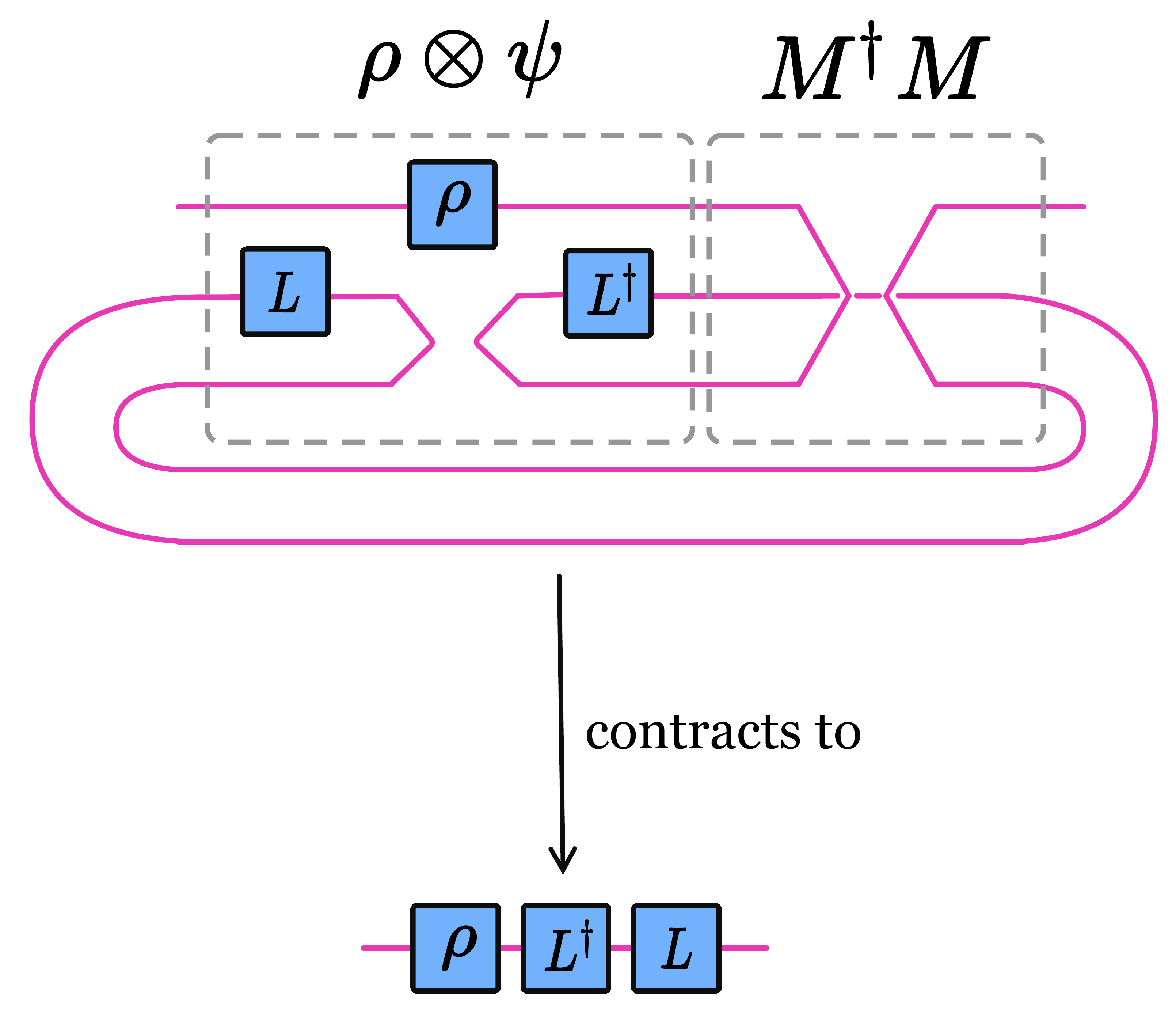}
    \caption{Tensor-network diagram of $\operatorname{Tr}_{23}\!
    \left [\left(\rho \otimes  \psi \right) M^{\dagger}M \right]$. Please refer to Figure~\ref{fig:lindblad-subroutine-no-H} to understand the tensor network diagram of $M^{\dagger}M$. The whole network on the top contracts to the network on the bottom due to the partial trace operation over the $2^{\text{nd}}$ and $3^{\text{rd}}$ systems. Simply follow the pink line from the left of the first system to its right to observe this.}
    \label{fig:lindblad-third-term-no-H}
\end{figure}

\subsection{Single Lindblad Operator Case with Hamiltonian Term}

\label{sec:single-op-with-H}

In this section, we consider the case of simulating a Lindbladian evolution, where the Lindbladian consists of a single Lindblad operator $L$ and a Hamiltonian $H$. To be more precise, we are interested in simulating the Lindbladian dynamics of a quantum state $\rho$ for time $t$ according to the following Lindbladian:
\begin{equation}\label{eq:lindbladmaster2-H}
    \mathcal{L}(\rho) \coloneqq -i [H, \rho] +  L\rho L^{\dagger} - \frac{1}{2} \left \{L^{\dagger}L, \rho \right\}.
\end{equation} 

To begin with, we will look at how to encode the operators $L$ and $H$ into quantum states. We suppose that  $L$ is encoded into a pure quantum state $\psi$ in the same way that was considered previously (see \eqref{eq:program-state-Lindblad-L}). On the other hand, we suppose that the Hamiltonian $H$ is encoded into the density matrix of a quantum state $\sigma$. This type of encoding was first considered in \cite{Lloyd2014QuantumAnalysis} for density matrix exponentiation, and further discussions of it are given in \cite[Eqs.~(1)--(2)]{Kimmel2017HamiltonianComplexity}. Overall, the program state that encodes $H$ and $L$ is the following tensor-product state:
\begin{equation}
    \omega \coloneqq \sigma \otimes \psi,
    \label{eq:omega-state-def}
\end{equation}
and the Lindbladian to be simulated, given by \eqref{eq:lindbladmaster2-H}, can now be rewritten as follows:
\begin{equation}\label{eq:lindbladmaster2}
    \mathcal{L}(\rho) \coloneqq -i [\sigma, \rho] +  L\rho L^{\dagger} - \frac{1}{2} \left \{L^{\dagger}L, \rho \right\}.
\end{equation}

We are now in a position to propose a quantum algorithm for simulating the quantum channel $e^{\mathcal{L}t}$, corresponding to the Lindbladian in \eqref{eq:lindbladmaster2}, up to error $\varepsilon$ in diamond distance, using $n$ copies of the program state $\omega$. As before, for providing an analysis related to the diamond distance, let $\rho \in \mathcal{D}(\mathcal{H}_{R} \otimes \mathcal{H}_{S})$ be an unknown quantum state given as input over the joint system $RS$, where the system $R$ acts as a reference system. Furthermore, let the $k^{\text{th}}$ copy of the program state $\omega$ be a quantum state of a joint system $H_{k}P_{k}Q_{k}$, where $\sigma \in \mathcal{D}(\mathcal{H}_{H_k})$ and $\psi \in \mathcal{D}(\mathcal{H}_{P_k} \otimes \mathcal{H}_{Q_k})$. For brevity, let us use $(HPQ)_{k}$ as a shorthand for~$H_{k}P_{k}Q_{k}$.

\textbf{Algorithm 2 ---} Set $n\in \mathbb{N}$, with a particular choice specified later. Set $k=1$. Given the $k^{\text{th}}$ copy of $\omega$, i.e., $\omega_{(HPQ)_{k}}$, perform the following two steps:
\begin{enumerate}
    \item Evolve the joint quantum state $\rho_{RS}\otimes \omega_{(HPQ)_{k}}$ according to the dynamics realized by the following Lindbladian $\mathcal{M}$, for some small duration of time $\Delta = t/n $:
    \begin{multline}\label{eq:general-fixed-L-master}
        \mathcal{M}\left (\rho_{RS} \otimes \omega_{(HPQ)_{k}}\right) \coloneqq -i [\hat{H}, \rho_{RS} \otimes  \omega_{(HPQ)_{k}} ] \\ + M(\rho_{RS} \otimes  \omega_{(HPQ)_{k}}) M^{\dagger} - \frac{1}{2} \left \{M^{\dagger}M, \rho_{RS} \otimes  \omega_{(HPQ)_{k}} \right\}.
    \end{multline}
 The Hamiltonian $\hat{H}$ and Lindblad operator~$M$ act on the joint system $RS(HPQ)_{k}$, and we define them as
    \begin{align}\label{eq:H+L-H}
        \hat{H} & \coloneqq (\mathsf{SWAP}_{SH_{k}} \otimes I),\\
        M & \coloneqq \frac{1}{\sqrt{d}} \left(I \otimes |\Gamma\rangle\!\langle\Gamma |_{P_{k}Q_{k}}\right) \left( \mathsf{SWAP}_{SP_{k}} \otimes I\right)\label{eq:H+L-L}.
    \end{align}
    Here, we apply the identity operator $I$ on all those systems that are not explicitly mentioned.
    \item Trace out the program states, i.e., the systems $(HPQ)_{k}$.
\end{enumerate}
We repeat the above procedure for each copy of $\omega$, i.e., for all $k$ ranging from 1 to $n$. As before, we have some flexibility in choosing $M$, as mentioned in Remark~\ref{rem:M-choice}.

The following theorem states that the above algorithm uses $n = O(t^2/\varepsilon)$ copies of $\omega$ to simulate the Lindbladian evolution of $\rho_{RS}$, according to the Lindbladian in \eqref{eq:lindbladmaster2}, for time $t$, and the resulting state is $\varepsilon$-close in normalized trace distance to the target state $(\mathcal{I}_{R} \otimes e^{\mathcal{L}t})(\rho_{RS})$, for every input state~$\rho_{RS}$.

\begin{theorem}\label{thm:single-op-plus-H}
 Given access to $n$ copies of the program state $\omega \in  \mathcal{D}\!\left(\mathcal{H}_{HPQ}\right)$, which is defined in~\eqref{eq:omega-state-def} and encodes the Lindblad operator $L$ and the Hamiltonian $H$, there exists a quantum algorithm $\mathcal{A}$ such that the following holds:
    \begin{equation}
        \frac{1}{2}\left \Vert e^{\mathcal{L}t} - \mathcal{A} \right \Vert_{\diamond} \leq \varepsilon,
    \end{equation}
    with only $n = O(t^2/\varepsilon)$ copies of $\omega$. In other words, $\mathcal{A}$ uses only $n = O(t^2/\varepsilon)$ copies of $\omega$ to approximate the channel $ e^{\mathcal{L}t}$, defined from \eqref{eq:lindbladmaster2} and \eqref{eq:Lind-expand}, up to $\varepsilon$ error in diamond distance.
\end{theorem}

 \begin{proof}
The proof style is very similar to that of Theorem~\ref{thm:single-op}. For clarity, we remove the system labels here as well.

\begin{figure}[t]
    \centering
    \includegraphics[scale=0.2]{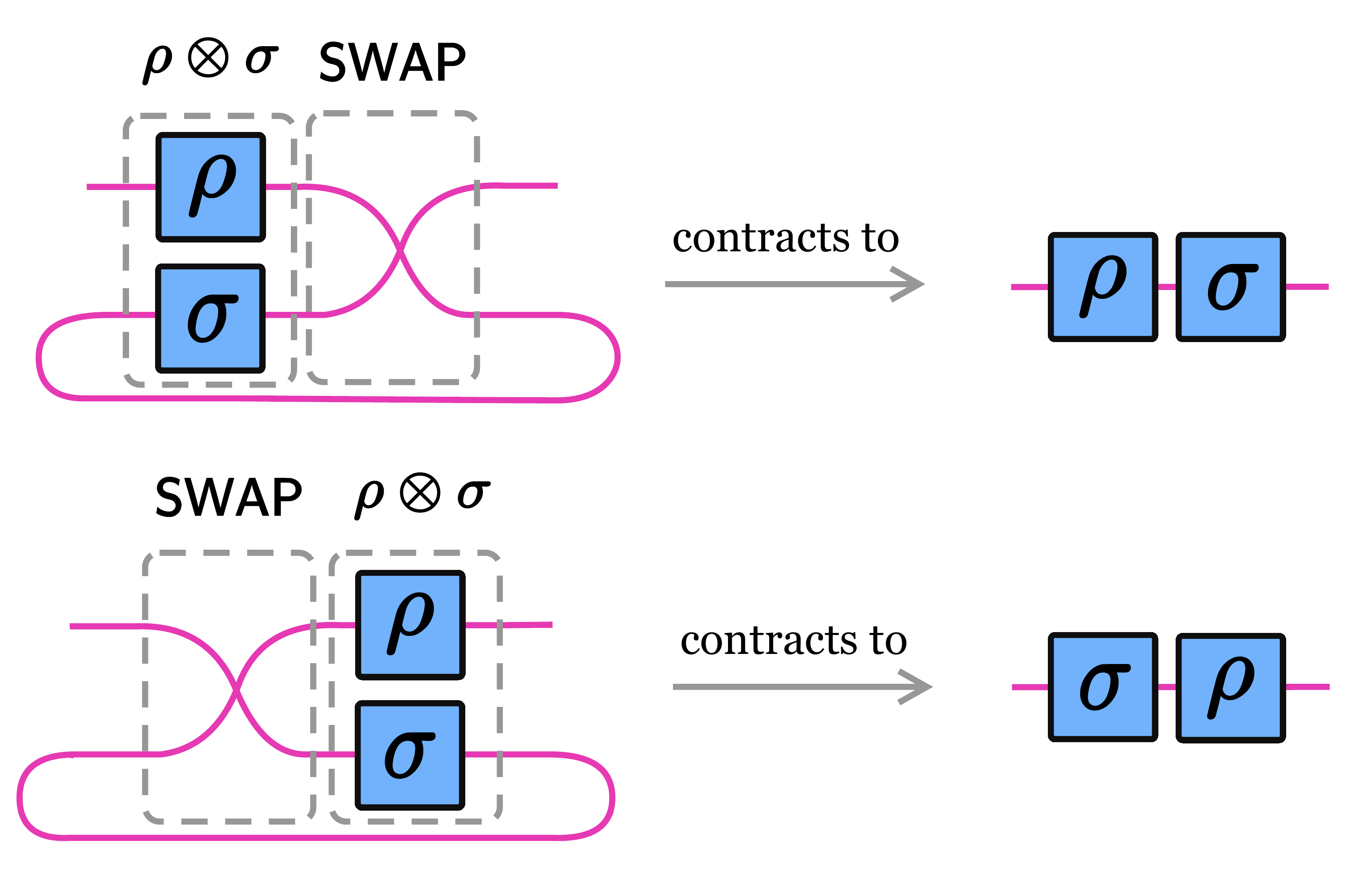}
    \caption{Tensor-network diagrams of $\operatorname{Tr}_{2} \left [\mathsf{SWAP} (\rho \otimes \sigma)\right]$ and $\operatorname{Tr}_{2} \left [(\rho \otimes \sigma) \mathsf{SWAP} \right]$. The networks on the left contract to the networks on the right due to the partial trace operation over the $2^{\text{nd}}$ systems. Simply follow the pink line from the left of the first system to its right to observe this.}
    \label{fig:lindblad-first-term-with-H}
\end{figure}

In the first step of Algorithm 2, we simulate the Lindbladian evolution of $\rho \otimes \omega$, given by \eqref{eq:lindbladmaster2}, with the Lindbladian $\mathcal{M}$ in \eqref{eq:general-fixed-L-master} for some small duration of time $\Delta$, and then trace out $\omega$. The output state obtained after this step is
\begin{equation}\label{eq:single-op-proof-sketch2}
    \operatorname{Tr}_{234}\!\left[e^{\mathcal{M}\Delta}(\rho \otimes \omega)\right] = \rho + \operatorname{Tr}_{234}\left[\mathcal{M}(\rho \otimes \omega) \right] \Delta + O(\Delta^2),
\end{equation}
where we have used the expansion in \eqref{eq:Lind-expand}.
Next we rewrite the second term in the following way:
\begin{multline}
    \operatorname{Tr}_{234}\!\left[\mathcal{M}\left (\rho \otimes \omega \right) \right ] = -i \operatorname{Tr}_{234}\!\left[ [\hat{H}, \rho \otimes  \omega ]\right] + \operatorname{Tr}_{234}\!\left [M(\rho \otimes \omega) M^{\dagger} \right]\\  - \frac{1}{2}\operatorname{Tr}_{234}\!
    \left [M^{\dagger}M \left(\rho \otimes  \omega\right) \right] - \frac{1}{2}\operatorname{Tr}_{234}\!
    \left [\left(\rho \otimes  \omega\right)M^{\dagger}M  \right].
\end{multline}
Observe that the operator $M$ can be written as $M = M' \otimes I_{2}$, where $I_{2}$ is the identity acting on the second register. From \eqref{eq:H+L-H} and \eqref{eq:H+L-L} and the fact that $\omega = \sigma \otimes \psi$, we simplify the right-hand side of the above equation:
\begin{multline}
    \operatorname{Tr}_{234}\!\left[\mathcal{M}\left (\rho \otimes \phi \right) \right ] = -i \operatorname{Tr}_{2}\!\left[ [\mathsf{SWAP}, \rho \otimes  \sigma ]\right] + \operatorname{Tr}_{34}\!\left [M'(\rho \otimes \psi) M'^{\dagger} \right]\\  - \frac{1}{2}\operatorname{Tr}_{34}\!
    \left [M'^{\dagger}M' \left(\rho \otimes  \psi\right) \right] - \frac{1}{2}\operatorname{Tr}_{34}\!
    \left [\left(\rho \otimes  \psi\right)M'^{\dagger}M'  \right].
\end{multline}
We simplify the first term even more by employing its related tensor-network diagrams, as illustrated in Figure \ref{fig:lindblad-first-term-with-H}. As before, please refer to Figures~\ref{fig:lindblad-first-term-no-H}, \ref{fig:lindblad-second-term-no-H}, and~\ref{fig:lindblad-third-term-no-H} to simplify the remaining terms (please keep in mind that $M$ is simply $M'$ when referring to these figures). We finally obtain the following equalities:
\begin{align}
-i \operatorname{Tr}_{2}\!\left[ [\mathsf{SWAP}, \rho \otimes  \sigma) ]\right] & = -i  [\sigma, \rho ],\\
    \operatorname{Tr}_{34}\!\left [M'(\rho \otimes \psi) M'^{\dagger} \right] & = L\rho L^{\dagger},\\
    \operatorname{Tr}_{34}\!
    \left [M'^{\dagger}M' \left(\rho \otimes  \psi \right) \right] & =  L^{\dagger}L \rho ,\\
    \operatorname{Tr}_{34}\!
    \left [\left(\rho \otimes  \psi\right)M'^{\dagger}M'  \right] & =   \rho L^{\dagger}L.
\end{align}%
%
From the above equations and \eqref{eq:lindbladmaster2}, we rewrite \eqref{eq:single-op-proof-sketch2} as
\begin{align}
\operatorname{Tr}_{234}\!\left[e^{\mathcal{M}\Delta}(\rho \otimes \omega)\right] & = \rho + \operatorname{Tr}_{234}\!\left[\mathcal{M}(\rho \otimes \omega) \right] \Delta + O(\Delta^2)\\
    & = \rho + \mathcal{L}(\rho) \Delta + O(\Delta^2)\\
    & = e^{\mathcal{L}\Delta}(\rho) + O(\Delta^2).
\end{align}

Substituting $\Delta = t/n$ and repeating Algorithm 2 for $n = O(t^2 / \varepsilon)$ times produces a quantum state that is $O(\varepsilon)$-close to the target  state~$e^{\mathcal{L}t}(\rho)$. A detailed error analysis of this claim goes along the lines of that provided in Appendix~\ref{app:error-analysis}.
\end{proof}

\section{Conclusion and Open Problems}

In this paper, we proposed a quantum algorithm for approximately simulating Lindblad evolution to arbitrary accuracy. For the purpose of this paper, we considered a simple case in which the Lindbladian consists of only one Lindblad operator, as this case can easily be extended to the more general case with multiple Lindblad operators. We further investigated the sample complexity of our algorithm for this case, i.e., the number of samples of the program state needed by our algorithm to achieve the desired accuracy. We then extended the single-operator case to include a Hamiltonian term in the Lindbladian, and we proposed a quantum algorithm for this case as well.

Here we list some directions for future work:
\begin{itemize}
\item Is there is an efficient implementation of Step 1 of Algorithm~1? Solving this problem will undoubtedly resolve the time or gate complexities of our algorithms. 

\item Another direction is to provide an extension to more complex cases, in which a Lindblad operator can be expressed as a linear combination or a polynomial of the operators encoded in the program states. This direction was considered in \cite[Section~5]{Kimmel2017HamiltonianComplexity} in the case of density matrix exponentiation / sample-based Hamiltonian simulation. 

\item In \cite{Kimmel2017HamiltonianComplexity}, the authors investigated a modified version of the Hamiltonian simulation problem, which they called sample-based Hamiltonian simulation. In this problem, given an unknown quantum state $\rho$ and $n$~copies of the program state $\sigma$, the task is to implement the following transformation:
\begin{equation}\label{eq:sample-based-hamil}
    \rho \otimes \underbrace{\sigma \otimes \cdots \otimes \sigma}_{\text{$n$ times}} \overset{\overset{\varepsilon}{\approx}}{\longrightarrow} e^{-i\sigma t}\rho e^{i\sigma t}.
\end{equation}
The authors reported in \cite[Theorem~5]{Kimmel2017HamiltonianComplexity} that a quantum algorithm for this task needs at least $\Omega(t^2/\varepsilon)$ copies of $\sigma$ to simulate the above channel within $\varepsilon$ accuracy. It is an important open question to determine the sample complexity of the general task of sample-based Lindbladian simulation, which could help determine if Algorithm~1 has optimal sample complexity. 

\item Furthermore, one can investigate the number of samples needed to perform state tomography of the program state and approximately recover the entire Lindblad operator encoded in this state. Then, we can simply use this operator to approximately simulate the corresponding Lindbladian evolution. One crucial question here is to ask if the sample complexity needed for state tomography is larger than the sample complexity of wave matrix Lindbladization. If this turns out to be the case, then it implies that a party can send sufficiently many copies of the program state to another party to simulate the respective Lindbladian evolution without revealing anything about the encoded Lindblad operator. 
One can think of this as some instance of quantum copy-protection, as introduced in~\cite{Aaronson2009}, in which the quantum operation is the Lindbladian evolution of a quantum state. 
This may also address the question raised in \cite{Kimmel2017HamiltonianComplexity}, which asks whether there is a quantum operation other than Hamiltonian simulation that can be encoded in quantum states and executed without revealing much about the quantum operation itself.

\end{itemize}

\section*{Acknowledgements}

DP and MMW thank Thomas Hollinger, Nana Liu, Māris Ozols, Aby Philip, Marina Radulaski, Vishal Singh, and Ewin Tang for insightful discussions. MMW is especially grateful to Prof.~Ingemar Bengtsson for the opportunity to have met and discussed research with Prof.~Lindblad in Stockholm, Sweden, during April~2019.

\appendix

\section{Proof of the Key Lemma}

\label{app:key-lemma}

\begin{lemma}\label{lemma:key-lemma}
    Let $\mathcal{H}_{1},\mathcal{H}_{2}$, and $\mathcal{H}_{3}$ be $d$-dimensional Hilbert spaces. Let $M$ be a linear operator acting on $\mathcal{H}_{1} \otimes \mathcal{H}_{2} \otimes \mathcal{H}_{3}$ and defined as
    \begin{equation}
    M_{123}\coloneqq \frac{1}{\sqrt{d}}\left(  I_{1}\otimes |\Gamma\rangle\! \langle \Gamma|_{23}\right)  \left(  \mathsf{SWAP}_{12}\otimes
    I_{3}\right).
    \end{equation}
    Also, let $\rho$ be a quantum state in $\mathcal{D}(\mathcal{H}_{1})$, and let $|\psi^{L}\rangle$ be a pure quantum state in $\mathcal{H}_{2}\otimes \mathcal{H}_{3}$ defined as $|\psi^{L}\rangle_{23}\coloneqq \left(  L_{2}\otimes I_{3}\right)  |\Gamma\rangle_{23},$ where $L$ is a $d \times d$-dimensional linear operator such that $\left\Vert L \right\Vert_{2} = 1$.
    Then the following identities hold:
    \begin{gather}\label{eq:lemma-first-identity}
        \operatorname{Tr}_{23}[M_{123}\left(  \rho_{1}\otimes\psi_{23}^{L}\right)
M_{123}^{\dag}] = \left[  L\rho L^{\dag}\right]  _{1},\\
\label{eq:lemma-second-identity}
\operatorname{Tr}_{23}[M_{123}^{\dag}M_{123}\left(  \rho_{1}\otimes
\psi_{23}^{L}\right)  ] = \left[  L^{\dag}L\rho\right]_{1},\\
\label{eq:lemma-third-identity}
\operatorname{Tr}_{23}[\left(  \rho_{1}\otimes\psi_{23}^{L}\right)
M_{123}^{\dag}M_{123}]=\left[  \rho L^{\dag}L\right]  _{1},
    \end{gather}
where we have used the shorthand $\psi_{23}^{L}\equiv|\psi^{L}\rangle\!\langle
\psi^{L}|_{23}$.
\end{lemma}
 \begin{proof}
Recall that
\begin{gather}
    |\Gamma\rangle\! \langle \Gamma|_{23} = \sum_{i,j}|i\rangle\! \langle j|_{2} \otimes|i\rangle\! \langle j|_{3},\\
    \mathsf{SWAP}_{12} = \sum_{k,\ell}|k\rangle\!
\langle\ell|_{1}\otimes|\ell\rangle\! \langle k|_{2}.
\end{gather}
Using the above equalites, observe that
\begin{align}
M_{123}  &  = \frac{1}{\sqrt{d}} \left(  I_{1}\otimes|\Gamma\rangle\! \langle \Gamma|_{23} \right)  \left(  \mathsf{SWAP}_{12}\otimes
I_{3}\right) \\
&  = \frac{1}{\sqrt{d}} \left(  I_{1}\otimes\sum_{i,j}|i\rangle\! \langle j|_{2}%
\otimes|i\rangle\! \langle j|_{3}\right)  \left(  \sum_{k,\ell}|k\rangle\!
\langle\ell|_{1}\otimes|\ell\rangle\! \langle k|_{2}\otimes I_{3}\right) \\
&  =\frac{1}{\sqrt{d}}\sum_{i,j,k,\ell}\left(  I_{1}\otimes|i\rangle\! \langle
j|_{2}\otimes|i\rangle\! \langle j|_{3}\right)  \left(  |k\rangle\! \langle
\ell|_{1}\otimes|\ell\rangle\! \langle k|_{2}\otimes I_{3}\right) \\
&  =\frac{1}{\sqrt{d}}\sum_{i,j,k,\ell}|k\rangle\! \langle\ell|_{1}\otimes
|i\rangle\!\langle j|\ell\rangle\! \langle k|_{2}\otimes|i\rangle\! \langle
j|_{3}\\
&  =\frac{1}{\sqrt{d}}\sum_{i,j,k}|k\rangle\! \langle j|_{1}\otimes|i\rangle\! \langle
k|_{2}\otimes|i\rangle\! \langle j|_{3}.%
\end{align}
Let us prove the first identity, i.e., that in \eqref{eq:lemma-first-identity}, as follows:
\begin{align}
&  \operatorname{Tr}_{23}\!\left[M_{123}\left(  \rho_{1}\otimes\psi_{23}^{L}\right)
M_{123}^{\dag}\right]\nonumber\\
&  =\operatorname{Tr}_{23}\Biggl[  \left(  \frac{1}{\sqrt{d}}\sum_{i,j,k}%
|k\rangle\! \langle j|_{1}\otimes|i\rangle\! \langle k|_{2}\otimes
|i\rangle\! \langle j|_{3}\right)  \left(  \rho_{1}\otimes\psi_{23}^{L}\right) \notag \times \\
& \hspace{4cm} \left(  \frac{1}{\sqrt{d}}\sum_{i^{\prime},j^{\prime},k^{\prime}}|j^{\prime}%
\rangle\! \langle k^{\prime}|_{1}\otimes|k^{\prime}\rangle\! \langle i^{\prime
}|_{2}\otimes|j^{\prime}\rangle\! \langle i^{\prime}|_{3}\right)  \Biggr]  \\
&  =\frac{1}{d}\sum_{i,j,k,i^{\prime},j^{\prime},k^{\prime}}%
\operatorname{Tr}_{23} \big[ \left(  |k\rangle\! \langle j|_{1}\otimes
|i\rangle\! \langle k|_{2}\otimes|i\rangle\! \langle j|_{3}\right)  \left(
\rho_{1}\otimes\psi_{23}^{L}\right) \times \notag \\
& \hspace{6cm} \left(  |j^{\prime}\rangle\! \langle
k^{\prime}|_{1}\otimes|k^{\prime}\rangle\! \langle i^{\prime}|_{2}%
\otimes|j^{\prime}\rangle\! \langle i^{\prime}|_{3}\right) \big]  \\ 
&  =\frac{1}{d}\sum_{i,j,k,i^{\prime},j^{\prime},k^{\prime}}%
|k\rangle\!\langle j|_{1}\rho_{1}|j^{\prime}\rangle\! \langle k^{\prime}%
|_{1}\operatorname{Tr}_{23}\!\left[  \left(  |i\rangle\! \langle k|_{2}%
\otimes|i\rangle\! \langle j|_{3}\right)  \psi_{23}^{L}\left(  |k^{\prime
}\rangle\! \langle i^{\prime}|_{2}\otimes|j^{\prime}\rangle\! \langle i^{\prime
}|_{3}\right)  \right]  \\
&  =\frac{1}{d}\sum_{i,j,k,i^{\prime},j^{\prime},k^{\prime}}%
|k\rangle\!\langle j|_{1}\rho_{1}|j^{\prime}\rangle\! \langle k^{\prime}%
|_{1}\ \left(  \langle k|_{2}\otimes\langle j|_{3}\right)  \psi_{23}%
^{L}\left(  |k^{\prime}\rangle_{2}\otimes|j^{\prime}\rangle_{3}\right)
\langle i^{\prime}|i\rangle_{2}\langle i^{\prime}|i\rangle_{3}\\
&  =\sum_{j,k,j^{\prime},k^{\prime}}|k\rangle\! \langle j|_{1}%
\rho_{1}|j^{\prime}\rangle\! \langle k^{\prime}|_{1}\ \left(  \langle
k|_{2}\otimes\langle j|_{3}\right)  \psi_{23}^{L}\left(  |k^{\prime}%
\rangle_{2}\otimes|j^{\prime}\rangle_{3}\right)  \\
&  =\sum_{j,k,j^{\prime},k^{\prime}}|k\rangle\! \langle j|_{1}%
\rho_{1}|j^{\prime}\rangle\! \langle k^{\prime}|_{1}\ \left(  \langle
k|_{2}\otimes\langle j|_{3}\right)  \left(  L_{2}\otimes I_{3}\right)
|\Gamma\rangle\! \langle\Gamma|_{23}\left(  L_{2}^{\dag}\otimes I_{3}\right)
 \notag \times \\
 & \hspace{8cm}\left(  |k^{\prime}\rangle_{2}\otimes|j^{\prime}\rangle_{3}\right)  \\
&  =\sum_{j,k,j^{\prime},k^{\prime},i,i^{\prime}}|k\rangle\! \langle
j|_{1}\rho_{1}|j^{\prime}\rangle\! \langle k^{\prime}|_{1}\ \left(  \langle
k|_{2}\otimes\langle j|_{3}\right)  \left(  L_{2}|i\rangle\! \langle i^{\prime
}|_{2}L_{2}^{\dag}\otimes|i\rangle\! \langle i^{\prime}|_{3}\right) \notag \times \\
& \hspace{8cm}\left(
|k^{\prime}\rangle_{2}\otimes|j^{\prime}\rangle_{3}\right)  \\
&  =\sum_{j,k,j^{\prime},k^{\prime},i,i^{\prime}}|k\rangle\! \langle
j|\rho|j^{\prime}\rangle\! \langle k^{\prime}|_{1}\ \langle k|L|i\rangle
\!\langle i^{\prime}|L^{\dag}|k^{\prime}\rangle\! \langle j|i\rangle\! \langle
i^{\prime}|j^{\prime}\rangle\\ 
&  =\sum_{j,k,j^{\prime},k^{\prime}}|k\rangle\! \langle
j|\rho|j^{\prime}\rangle\! \langle k^{\prime}|_{1}\ \langle k|L|j\rangle
\!\langle j^{\prime}|L^{\dag}|k^{\prime}\rangle\\
&  =\sum_{j,k,j^{\prime},k^{\prime}}|k\rangle\! \langle
k|L|j\rangle\!\langle j|\rho|j^{\prime}\rangle\! \langle j^{\prime}|L^{\dag
}|k^{\prime}\rangle\! \langle k^{\prime}|_{1}\\
&  =\left[  L\rho L^{\dag}\right]  _{1}.
\end{align}
Now consider that
\begin{align}
M_{123}^{\dag}M_{123} &  =\left(  \frac{1}{\sqrt{d}}\sum_{i^{\prime},j^{\prime
},k^{\prime}}|j^{\prime}\rangle\! \langle k^{\prime}|_{1}\otimes|k^{\prime
}\rangle\! \langle i^{\prime}|_{2}\otimes|j^{\prime}\rangle\! \langle i^{\prime
}|_{3}\right) \times \notag  \\
& \qquad \qquad \left(  \frac{1}{\sqrt{d}}\sum_{i,j,k}|k\rangle\! \langle j|_{1}%
\otimes|i\rangle\! \langle k|_{2}\otimes|i\rangle\! \langle j|_{3}\right)  \\
&  =\frac{1}{d}\sum_{i^{\prime},j^{\prime},k^{\prime},i,j,k}|j^{\prime
}\rangle\! \langle k^{\prime}|k\rangle\! \langle j|_{1}\otimes|k^{\prime}%
\rangle\! \langle i^{\prime}|i\rangle\! \langle k|_{2}\otimes|j^{\prime}%
\rangle\! \langle i^{\prime}|i\rangle\! \langle j|_{3}\\
&  =\sum_{j^{\prime},j,k}|j^{\prime}\rangle\! \langle j|_{1}%
\otimes|k\rangle\! \langle k|_{2}\otimes|j^{\prime}\rangle\! \langle j|_{3}\\
&  =\sum_{j^{\prime},j}|j^{\prime}\rangle\! \langle j|_{1}\otimes
I_{2}\otimes|j^{\prime}\rangle\! \langle j|_{3}\\
&  =\sum_{i,j}|i\rangle\! \langle j|_{1}\otimes I_{2}\otimes
|i\rangle\! \langle j|_{3} .
\end{align}
Then, for checking the second identity, i.e., that in \eqref{eq:lemma-second-identity}, we find that%
\begin{align}
&  \operatorname{Tr}_{23}\!\left[M_{123}^{\dag}M_{123}\left(  \rho_{1}\otimes
\psi_{23}^{L}\right)  \right]\nonumber\\
&  =\operatorname{Tr}_{23}\!\left[  \left(\sum_{i,j}|i\rangle
\!\langle j|_{1}\otimes I_{2}\otimes|i\rangle\! \langle j|_{3}\right)  \left(
\rho_{1}\otimes\psi_{23}^{L}\right)  \right]  \\
&  =\sum_{i,j}\operatorname{Tr}_{23}\!\left[  \left(  |i\rangle
\!\langle j|_{1}\otimes I_{2}\otimes|i\rangle\! \langle j|_{3}\right)  \left(
\rho_{1}\otimes\psi_{23}^{L}\right)  \right]  \\
&  =\sum_{i,j}|i\rangle\! \langle j|_{1}\rho_{1}\ \operatorname{Tr}%
_{23}\left[  \left(  I_{2}\otimes|i\rangle\! \langle j|_{3}\right)  \left(
\psi_{23}^{L}\right)  \right]  \\
&  =\sum_{i,j}|i\rangle\! \langle j|_{1}\rho_{1}\ \operatorname{Tr}%
_{23}\left[  \left(  I_{2}\otimes|i\rangle\! \langle j|_{3}\right)  \left(
L_{2}\otimes I_{3}\right)  |\Gamma\rangle\! \langle\Gamma|_{23}\left(
L_{2}^{\dag}\otimes I_{3}\right)  \right]  \\ 
&  =\sum_{i,j,k,\ell}|i\rangle\! \langle j|_{1}\rho_{1}%
\ \operatorname{Tr}_{23}\!\left[  \left(  I_{2}\otimes|i\rangle\! \langle
j|_{3}\right)  \left(  L_{2}\otimes I_{3}\right)  \left(  |k\rangle
\!\langle\ell|_{2}\otimes|k\rangle\! \langle\ell|_{3}\right)  \left(
L_{2}^{\dag}\otimes I_{3}\right)  \right]  \\
&  =\sum_{i,j,k,\ell}|i\rangle\! \langle j|_{1}\rho_{1}%
\ \operatorname{Tr}_{23}\!\left[  L_{2}|k\rangle\! \langle\ell|_{2}L_{2}^{\dag
}\otimes|i\rangle\! \langle j|k\rangle\! \langle\ell|_{3}\right]  \\
&  =\sum_{i,j,k,\ell}|i\rangle\! \langle j|_{1}\rho_{1}%
\ \operatorname{Tr}[L_{2}|k\rangle\! \langle\ell|_{2}L_{2}^{\dag}%
]\operatorname{Tr}[|i\rangle\! \langle j|k\rangle\! \langle\ell|_{3}]\\
&  =\sum_{i,j,k,\ell}|i\rangle\! \langle j|_{1}\rho_{1}\ \langle
\ell|L^{\dag}L|k\rangle\! \langle\ell|i\rangle\! \langle j|k\rangle\\
&  =\sum_{i,j}|i\rangle\! \langle j|_{1}\rho_{1}\ \langle i|L^{\dag
}L|j\rangle\\
&  =\sum_{i,j}|i\rangle\! \langle i|L^{\dag}L|j\rangle\! \langle
j|\rho\\
&  =\left[  L^{\dag}L\rho\right]  _{1}.
\end{align}
Then the third identity in \eqref{eq:lemma-third-identity} is the Hermitian conjugate of the above identity, and so we find that
\begin{equation}
\operatorname{Tr}_{23}\!\left[\left(  \rho_{1}\otimes\psi_{23}^{L}\right)
M_{123}^{\dag}M_{123}\right]=\left[  \rho L^{\dag}L\right]_{1}.
\end{equation}
This concludes the proof.
\end{proof}

\section{Proof of Theorem~\ref{thm:single-op}}

\label{app:error-analysis}

The error analysis here has some similarities with that from \cite[Appendix~B]{Kimmel2017HamiltonianComplexity}. We first expand the target state using a Taylor series expansion, as in \eqref{eq:Lind-expand}:
\begin{multline}
\label{eq:single-op-thm-proof-target}
    (\mathcal{I}_{R} \otimes e^{\mathcal{L}t})(\rho_{RS}) = \rho_{RS} + (\mathcal{I}_{R} \otimes \mathcal{L})(\rho_{RS}) t \\
    + \frac{1}{2} ((\mathcal{I}_{R} \otimes \mathcal{L})\circ (\mathcal{I}_{R} \otimes \mathcal{L}))(\rho_{RS}) t^2 + \ldots \quad.
\end{multline}
For the sake of brevity, we will no longer explicitly state the action of the identity channel $\mathcal{I}$ on the system $R$, and we leave it implicit that $\mathcal{L}$ and $L$ act on system $S$ alone. 

Next, let $\rho_{RS}^{(1)}$ be the output quantum state obtained after the first step of Algorithm~1. Formally, it can be written as follows:
\begin{align}
    \rho_{RS}^{(1)} \notag & \coloneqq \operatorname{Tr}_{P
_{1}Q_{1}}\!\left [ e^{\mathcal{M}\Delta}\left(\rho_{RS}\otimes \psi_{P_{1}Q_{1}}\right)\right ]\\
    & = \rho_{RS} 
+ \operatorname{Tr}_{P
_{1}Q_{1}}\!\left [M(\rho_{RS} \otimes  \psi_{P_{1}Q_{1}}) M^{\dagger} - \frac{1}{2} \left \{M^{\dagger}M, \rho_{RS} \otimes  \psi_{P_{1}Q_{1}} \right\} \right] \Delta \notag \\
    & \qquad + \frac{1}{2} \operatorname{Tr}_{P
_{1}Q_{1}}\!\left [(\mathcal{M}\circ \mathcal{M}) \left (\rho_{RS} \otimes  \psi_{P_{1}Q_{1}}\right)\right] \Delta^2 + O(\Delta^3),
\end{align}
where the equality follows by expanding $e^{\mathcal{M}\Delta}\left(\rho_{RS}\otimes \psi_{P_{1}Q_{1}}\right)$ at $\Delta=0$ using a Taylor series expansion, as in \eqref{eq:Lind-expand}.

We focus now on simplifying the second term above (the one proportional to~$\Delta$). We begin by expanding and rewriting it as follows:
\begin{multline}
    \operatorname{Tr}_{P_{1}Q_{1}} \!\left [M(\rho_{RS} \otimes \psi_{P_{1}Q_{1}}) M^{\dagger} \right] - \frac{1}{2}\operatorname{Tr}_{P_{1}Q_{1}}\!
    \left [M^{\dagger}M \left(\rho_{RS} \otimes  \psi_{P_{1}Q_{1}}\right) \right]\\
    - \frac{1}{2}\operatorname{Tr}_{P_{1}Q_{1}}
    \!\left [\left(\rho_{RS} \otimes  \psi_{P_{1}Q_{1}}\right)M^{\dagger}M  \right].
    \label{eq:2nd-term-to-simplify}
\end{multline}
We then invoke Lemma~\ref{lemma:key-lemma} to simplify each of the above terms. As a result of this, we obtain
\begin{align}
    \operatorname{Tr}_{P_{1}Q_{1}} \!\left [M(\rho_{RS} \otimes \psi_{P_{1}Q_{1}}) M^{\dagger} \right] & = L\rho_{RS}L^{\dagger}, \\
    \operatorname{Tr}_{P_{1}Q_{1}}\! \left [M^{\dagger}M \left(\rho_{RS} \otimes  \psi_{P_{1}Q_{1}}\right) \right] & =  L^{\dagger}L \rho_{RS} , \\
    \operatorname{Tr}_{P_{1}Q_{1}}\! \left [\left(\rho_{RS} \otimes  \psi_{P_{1}Q_{1}}\right)M^{\dagger}M  \right] & =   \rho_{RS}L^{\dagger}L.
\end{align}
For visualizing the above simplifications graphically, please refer to the tensor-network diagrams provided in Figures~\ref{fig:lindblad-first-term-no-H}, \ref{fig:lindblad-second-term-no-H}, and \ref{fig:lindblad-third-term-no-H}, respectively.

Using the above equations, the expression in \eqref{eq:2nd-term-to-simplify} above can be written as
\begin{multline}
\operatorname{Tr}_{P
_{1}Q_{1}}\!\left [M(\rho_{RS} \otimes  \psi_{P_{1}Q_{1}}) M^{\dagger} - \frac{1}{2} \left \{M^{\dagger}M, \rho_{RS} \otimes  \psi_{P_{1}Q_{1}} \right\} \right]   \\ = L\rho_{RS} L^{\dagger} - \frac{1}{2} \left \{L^{\dagger}L, \rho_{RS} \right\} = \mathcal{L}(\rho_{RS}).
\end{multline}

Therefore, the output quantum state obtained after the first iteration of Algorithm~1 is
 \begin{multline}
 \label{eq:single-rho-after-first-iter}
     \rho_{RS}^{(1)} = \rho_{RS} + \mathcal{L}(\rho_{RS}) \Delta 
     +  \frac{1}{2} \operatorname{Tr}_{P
_{1}Q_{1}}\!\left [(\mathcal{M}\circ \mathcal{M}) \left (\rho_{RS} \otimes  \psi_{P_{1}Q_{1}}\right)\right] \Delta^2 + O(\Delta^3).
 \end{multline}
For simplicity, let us use the following shorthand for the third term of~\eqref{eq:single-rho-after-first-iter} (i.e., the term proportional to $\Delta^2$):
\begin{equation}
    \mathcal{J}(\rho_{RS}) \coloneqq \operatorname{Tr}_{P
_{i}Q_{i}}\!\left [(\mathcal{M}\circ \mathcal{M}) \left (\rho_{RS} \otimes  \psi_{P_{i}Q_{i}}\right)\right] \Delta^2, \text{  for all } i\in \{1, \ldots, n\}.
\end{equation}
Substituting the above equation into \eqref{eq:single-rho-after-first-iter}, we get
\begin{equation}
    \rho_{RS}^{(1)} = \rho_{RS} + \mathcal{L}(\rho_{RS}) \Delta + \frac{1}{2}\mathcal{J}(\rho_{RS}) \Delta^2 + O(\Delta^3).
\end{equation}

Following that, we use the above development to obtain the output state after $k$ iterations of our quantum algorithm, i.e., 
\begin{equation}
    \rho^{(k)} = \rho^{(k-1)} + \mathcal{L}(\rho^{(k-1)}) \Delta +  \frac{1}{2}\mathcal{J}(\rho^{(k-1)}) \Delta^2 + O(\Delta^3),
\end{equation}
where here and in what follows we  avoid showing system labels for simplicity. 
Then considering that%
\begin{equation}
\rho^{(k-1)}=\rho^{(k-2)}+\mathcal{L}(\rho^{(k-2)})\Delta+\frac{1}%
{2}\mathcal{J}(\rho^{(k-2)})\Delta^{2}+O(\Delta^{3}),
\end{equation}
we find that%
\begin{equation}
\rho^{(k)}=\rho^{(k-2)}+\mathcal{L}(\rho^{(k-2)})2\Delta+\frac{1}%
{2}\mathcal{J}(\rho^{(k-2)})2\Delta^{2}+\mathcal{L}^{2}(\rho^{(k-2)}%
)\Delta^{2}+O(\Delta^{3}),
\end{equation}
because%
\begin{align}
& \rho^{(k)} \notag \\
&  =\rho^{(k-1)}+\mathcal{L}(\rho^{(k-1)})\Delta+\frac{1}%
{2}\mathcal{J}(\rho^{(k-1)})\Delta^{2}+O(\Delta^{3})\\
&  =\rho^{(k-2)}+\mathcal{L}(\rho^{(k-2)})\Delta+\frac{1}{2}\mathcal{J}%
(\rho^{(k-2)})\Delta^{2}+O(\Delta^{3})\nonumber\\
&  \qquad+\mathcal{L}\left(  \rho^{(k-2)}+\mathcal{L}(\rho^{(k-2)}%
)\Delta+\frac{1}{2}\mathcal{J}(\rho^{(k-2)})\Delta^{2}+O(\Delta^{3})\right)
\Delta\nonumber\\
&  \qquad+\frac{1}{2}\mathcal{J}\left(  \rho^{(k-2)}+\mathcal{L}(\rho
^{(k-2)})\Delta+\frac{1}{2}\mathcal{J}(\rho^{(k-2)})\Delta^{2}+O(\Delta
^{3})\right)  \Delta^{2}\\
&  =\rho^{(k-2)}+\mathcal{L}(\rho^{(k-2)})\Delta+\frac{1}{2}\mathcal{J}%
(\rho^{(k-2)})\Delta^{2}\nonumber\\
&  \qquad+\mathcal{L}\left(  \rho^{(k-2)}\right)  \Delta+\mathcal{L}^{2}%
(\rho^{(k-2)})\Delta^{2}+\frac{1}{2}\mathcal{J}\left(  \rho^{(k-2)}\right)
\Delta^{2}+O(\Delta^{3})\\
&  =\rho^{(k-2)}+\mathcal{L}(\rho^{(k-2)})2\Delta+\frac{1}{2}\mathcal{J}%
(\rho^{(k-2)})2\Delta^{2}+\mathcal{L}^{2}(\rho^{(k-2)})\Delta^{2}+O(\Delta
^{3}).
\end{align}
Repeating this kind of analysis several times leads to
\begin{align}
\rho^{(k)}& = \rho^{(k-1)} + \mathcal{L}(\rho^{(k-1)}) \Delta +  \frac{1}{2}\mathcal{J}(\rho^{(k-1)}) \Delta^2 + O(\Delta^3) \notag \\
& =  \rho^{(k-2)}+\mathcal{L}(\rho^{(k-2)})2\Delta+\frac{1}%
{2}\mathcal{J}(\rho^{(k-2)})2\Delta^{2}+\mathcal{L}^{2}(\rho^{(k-2)}%
)\Delta^{2}+O(\Delta^{3}) \notag  \\
& =\rho^{(k-3)}+\mathcal{L}(\rho^{(k-3)})3\Delta+\frac{1}%
{2}\mathcal{J}(\rho^{(k-3)})3\Delta^{2}+\mathcal{L}^{2}(\rho^{(k-3)}%
)3\Delta^{2}+O(\Delta^{3})\notag \\
& = \rho^{(k-4)}+\mathcal{L}(\rho^{(k-4)})4\Delta+\frac{1}%
{2}\mathcal{J}(\rho^{(k-4)})4\Delta^{2}+\mathcal{L}^{2}(\rho^{(k-4)}%
)6\Delta^{2}+O(\Delta^{3}) \notag \\
& = \rho^{(k-5)}+\mathcal{L}(\rho^{(k-5)})5\Delta+\frac{1}%
{2}\mathcal{J}(\rho^{(k-5)})5\Delta^{2}+\mathcal{L}^{2}(\rho^{(k-5)}%
)10\Delta^{2}+O(\Delta^{3}).
\label{eq:base-step-induction}
\end{align}
As such, the above establishes the base step for a proof by induction. To see the inductive step, suppose that
\begin{align}
\rho^{(k)} &  =\rho^{(k-m)}+\mathcal{L}(\rho^{(k-m)})m\Delta+\frac{1}%
{2}\mathcal{J}(\rho^{(k-m)})m\Delta^{2}\notag \\
& \qquad\qquad +\mathcal{L}^{2}(\rho^{(k-m)})\left(
1+2+\cdots+m-1\right)  \Delta^{2}+O(\Delta^{3})
\label{eq:inductive-assumption-1}\\
&  =\rho^{(k-m)}+\mathcal{L}(\rho^{(k-m)})m\Delta+\frac{1}{2}\mathcal{J}%
(\rho^{(k-m)})m\Delta^{2}\notag \\
& \qquad\qquad +\mathcal{L}^{2}(\rho^{(k-m)})\frac{m\left(
m-1\right)  }{2}\Delta^{2}+O(\Delta^{3}) \label{eq:inductive-assumption},
\end{align}
as is consistent with \eqref{eq:base-step-induction}.
Then, by plugging the following into \eqref{eq:inductive-assumption-1}
\begin{equation}
\rho^{(k-m)}=\rho^{(k-m-1)}+\mathcal{L}(\rho^{(k-m-1)})\Delta+\frac{1}%
{2}\mathcal{J}(\rho^{(k-m-1)})\Delta^{2}+O(\Delta^{3}),
\end{equation}
we find that%
\begin{align}
&  \rho^{(k)}\nonumber\\
&  =\rho^{(k-m)}+\mathcal{L}(\rho^{(k-m)})m\Delta+\frac{1}{2}\mathcal{J}%
(\rho^{(k-m)})m\Delta^{2}\nonumber\\
&  \qquad+\mathcal{L}^{2}(\rho^{(k-m)})\left(  1+\cdots+m-1\right)  \Delta
^{2}+O(\Delta^{3})\\
&  =\rho^{(k-m-1)}+\mathcal{L}(\rho^{(k-m-1)})\Delta+\frac{1}{2}%
\mathcal{J}(\rho^{(k-m-1)})\Delta^{2}+O(\Delta^{3})\nonumber\\
&  \qquad+\mathcal{L}\left(  \rho^{(k-m-1)}+\mathcal{L}(\rho^{(k-m-1)}%
)\Delta+\frac{1}{2}\mathcal{J}(\rho^{(k-m-1)})\Delta^{2}+O(\Delta^{3})\right)
m\Delta\nonumber\\
&  \qquad+\frac{1}{2}\mathcal{J}\left(  \rho^{(k-m-1)}+\mathcal{L}%
(\rho^{(k-m-1)})\Delta+\frac{1}{2}\mathcal{J}(\rho^{(k-m-1)})\Delta
^{2}+O(\Delta^{3})\right)  m\Delta^{2}\nonumber\\
&  \qquad+\mathcal{L}^{2}\left(
\begin{array}
[c]{c}%
\rho^{(k-m-1)}+\mathcal{L}(\rho^{(k-m-1)})\Delta\\
+\frac{1}{2}\mathcal{J}(\rho^{(k-m-1)})\Delta^{2}+O(\Delta^{3})
\end{array}
\right)  \left(  1+\cdots+m-1\right)  \Delta^{2}\\
&  =\rho^{(k-m-1)}+\mathcal{L}(\rho^{(k-m-1)})\Delta+\frac{1}{2}%
\mathcal{J}(\rho^{(k-m-1)})\Delta^{2}+O(\Delta^{3})\nonumber\\
&  \qquad+\mathcal{L}\left(  \rho^{(k-m-1)}\right)  m\Delta+\mathcal{L}%
^{2}(\rho^{(k-m-1)})m\Delta^{2}\nonumber\\
&  \qquad+\frac{1}{2}\mathcal{J}\left(  \rho^{(k-m-1)}\right)  m\Delta
^{2}+\mathcal{L}^{2}\left(  \rho^{(k-m-1)}\right)  \left(  1+\cdots
+m-1\right)  \Delta^{2}\\
&  =\rho^{(k-m-1)}+\mathcal{L}(\rho^{(k-m-1)})\left(  m+1\right)  \Delta
+\frac{1}{2}\mathcal{J}(\rho^{(k-m-1)})\left(  m+1\right)  \Delta
^{2}\nonumber\\
&  \qquad+\mathcal{L}^{2}\left(  \rho^{(k-m-1)}\right)  \left(  1+\cdots
+m\right)  \Delta^{2}+O(\Delta^{3}).
\end{align}
This inductive proof thus establishes the claimed formula in \eqref{eq:inductive-assumption}.


By setting $k = m = n$ and plugging into \eqref{eq:inductive-assumption}, we obtain an expression for the output state after the $n^{\text{th}}$ iteration, in terms of the input state, i.e., $\rho_{RS}^{(0)} = \rho_{RS}$,
\begin{equation}
    \rho_{RS}^{(n)} = \rho_{RS} +  \mathcal{L}(\rho_{RS}) n\Delta +  \frac{1}{2} \mathcal{J}(\rho_{RS}) n \Delta^2 +   (\mathcal{L} \circ \mathcal{L})(\rho_{RS}) \frac{n (n-1)}{2}\Delta^2 + O(\Delta^3).
\end{equation}
Substituting $t = n\Delta$ into the  target state given by \eqref{eq:single-op-thm-proof-target} and comparing it with the above state, we get
\begin{align}
    & \frac{1}{2}\left \Vert \rho_{RS}^{(n)} -  e^{\mathcal{L}n \Delta}(\rho_{RS})  \right\Vert_{1} \notag \\
    & = \frac{1}{4} \left \Vert \mathcal{J}(\rho_{RS}) n \Delta^2 -   (\mathcal{L} \circ \mathcal{L})(\rho_{RS}) n\Delta^2 + O(\Delta^3)\right \Vert_{1}\\
    & \leq O(n \Delta^2).
\end{align}
This implies that in order to achieve the desired accuracy of $O(\varepsilon)$, we need to repeat Algorithm~1 $n = O(t^2/\varepsilon)$ times, and for this, we need $n = O(t^2/\varepsilon)$ copies of the program state $\psi$. As the aforementioned argument holds for an arbitrary input state $\rho_{RS}$, the error bound holds more generally for the diamond distance.

\bibliography{references}
\bibliographystyle{plain}

\end{document}